\documentclass[conference]{IEEEtran}
\usepackage{amsmath,amsthm,amssymb,graphicx,graphics,epsfig,color}
\usepackage{lmodern}
\RequirePackage{fix-cm}

\allowdisplaybreaks

\def\A{\mathcal{A}}
\def\P{\mathcal{P}_n}
\def\L{\mathcal{L}_D(P)}
\def\hL{\mathcal{L}_D}
\def\sfL{\mathsf{L}}
\def\le#1{l^{(#1)}}
\def\p1p{(p_1,p_2,\ldots,p_{n})}
\def\level{level property }
\def\maxsib{maximum-level sibling property }
\def\sib{sibling property }

\newtheorem{thm}{Theorem}
 \newtheorem{lemma}{Lemma}

\newtheorem{remark}{Remark}

\begin{document}

\title{Probability Mass Functions for which Sources have the Maximum Minimum Expected Length}

\author{\IEEEauthorblockN{Shivkumar~K.~Manickam}
\IEEEauthorblockA{Dept. of Electrical Communication Engineering\\
Indian Institute of Science, Bangalore, India\\
Email: shivkumar@iisc.ac.in}}

\maketitle

\begin{abstract}
 Let $\P$ be the set of all probability mass functions (PMFs) $(p_1,p_2,\ldots,p_n)$ that satisfy $p_i>0$ for $1\leq i \leq n$.
Define the minimum expected length function $\hL:\P\rightarrow \mathbb{R}$ such that $\hL(P)$ is the minimum expected length of a prefix code, formed out of an alphabet of size $D$, for the discrete memoryless source having $P$ as its source distribution. It is well-known that the function $\hL$ attains its maximum value at the uniform distribution. Further, when $n$ is of the form $D^m$, with $m$ being a positive integer, PMFs other than the uniform distribution at which $\hL$ attains its maximum value are known. However, a complete characterization of all such PMFs at which the minimum expected length function attains its maximum value has not been done so far. This is done in this paper. 
%
\end{abstract}
\section{Introduction to the Problem}
One of the earliest problems considered in information theory is that of finding a prefix code with the minimum expected length 
for a given discrete memoryless source. 
This paper addresses a question related to this problem.

Let us begin by establishing the basic terminology and notation used in this paper. 
A set of finite length strings of letters coming from a given finite alphabet is said to be a \emph{prefix code} if no string is a prefix of another.
 Let us use $\A$ to denote the finite alphabet and $D$ to denote its size. 
The members of a prefix code are called
 \emph{codewords}. 
 
 Consider a discrete source with $n$ symbols 
 where the $i$th symbol occurs with a probability $p_i$ ($\sum_i p_i =1$). The collection of probabilities $P=(p_1,p_2,\ldots,p_n)$ is a \emph{probability mass function} (PMF). 
 Let $X$ be a prefix code assigned to this source (from now onwards, we will simply say $X$ is a prefix code for the PMF $P$, leaving out any mention of the source) and let the $i$th source symbol be associated with a codeword of length $l_i$. The \emph{expected length} of the prefix code $X$ is $\sum_i p_il_i$. 
 A minimum expected length prefix code can be effectively obtained using the Huffman algorithm \cite{Huffman}. 
  Henceforth, we will refer to a minimum expected length prefix code as an
  \emph{optimal code}. 
 
 Let $\P$ be the set of all PMFs $(p_1,p_2,\ldots,p_n)$ that satisfy $p_i>0$ for $1\leq i \leq n$. Define the function $\hL:\P\rightarrow \mathbb{R}$ such that $\hL(P)$ is the expected length of an optimal code for 
 the PMF $P$.
  Let us call this function the \emph{minimum expected length function}.
 
 Now, for a PMF $P\in \P$, the only known general way to determine $\hL(P)$ is by first determining an optimal code using the Huffman algorithm and then finding its expected length. There is neither any known analytical formula for $\L$ in terms of the probabilities of $P$ nor an alternate characterization of the function $\hL$, from which its values can be readily evaluated. However, some properties of this function are known.  Let us call a PMF at which $\hL$ attains its maximum value to be a \emph{point of maximum}. A result of Hwang \cite{Hwang} shows that the minimum expected length function is Schur-concave, and so attains its maximum value at the uniform distribution: $U_n=(1/n,1/n,\ldots,1/n)$.
Further, when $n$ is of the form $D^m$, a result of \cite{extra} gives other points of maximum:
all $P\in \P$ in which the sum of the lowest $D$ probabilities 
is greater than or equal to the highest probability (we will see more about this).
 However, to the author's knowledge, a complete characterization of all the 
 points of maximum
 has not been done so far. This is carried out in this paper. 
 

 We will be making use of a characterization of Huffman trees given by Gallager \cite{Gallager}. This is presented in the next section.
%

\section{Huffman Trees}
It is useful to visualize prefix codes in the form of trees  \cite[Chapter~5]{Cover_Thomas} (see Fig.~\ref{codetree}). 
For our purpose, we will find it convenient to give directions to the edges of a tree. When we will refer to a directed graph as a tree, we will do so in the sense that the undirected graph obtained by replacing each directed edge by an undirected one is a tree. 
Consider an infinite rooted directed tree in which all the edges are directed away from the root. Let each node of the tree have $D$ outgoing edges.
Let us denote this tree as $T_{\infty}$. Throughout this paper, whenever we talk of a tree we will mean a subtree of 
$T_{\infty}$.
In a tree, if there is an edge from a node $v$ to node $v_1$, then $v$ is said to be the \emph{parent} of $v_1$ and $v_1$ is said to be a \emph{child} of $v$. A node $d$ is said to be a \emph{descendant} of a node $v$ if there is a path from $v$ to $d$. Nodes having the same parent are called \emph{siblings}. A \emph{sibling set} is the set of the children of an internal node of a tree. The \emph{level} of a node $v$ is the length of the path from the root to $v$.

\begin{figure}
  \centering
   \resizebox{5.8cm}{!}{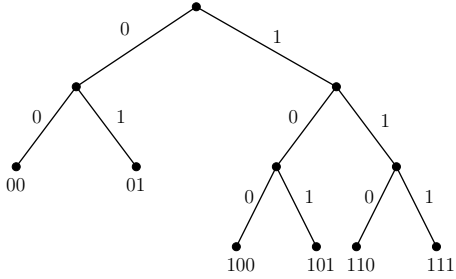}
    \caption{Code tree for the prefix-free code $\{00,01,100,101,110,111 \}$.} 
\label{codetree}
\end{figure}

For every node of $T_{\infty}$, label each of the outgoing edges with a letter from $\A$ such that no two edges is associated with the same letter. Associate each node with the string formed by reading out the labels in the path from the root to that particular node. 
To get the tree representation of a prefix code, mark all the nodes of the labeled tree $T_{\infty}$ that correspond to the codewords. Retain all the nodes in the paths from the root to the marked nodes (including the root and the marked nodes) and delete the remaining nodes. The resulting tree is called a \emph{code tree}.
Conversely, any tree $T$ can be considered to represent a class of prefix codes. To get this class,
consider all the possible ways the edges of the tree $T$ can be labeled using the letters of $\A$ (with no two outgoing edges from a node getting the same label). From each assignment of labels on the edges, we can get a prefix code by collecting all the strings along the leaves of $T$. Let us consider $T$ to represent the class of all such prefix codes obtained from all the different assignments of labels. 

Let $T$ be a tree with $n$ leaves and let $P\in \P$. If there is a 1-1 mapping between the probabilities of $P$ and the leaves of $T$, we say that $T$ is a tree associated with the PMF $P$. Using our relation between a tree and a class of prefix codes, we can see that a tree $T$ associated with a PMF $P$ defines a class of prefix codes for $P$. Observe that the expected length of all these prefix codes is the same, and we take this value as the \emph{expected length} of the tree $T$. An \emph{optimal tree} is that having the minimum expected length. The Huffman algorithm can be viewed as a one that constructs an optimal tree. 
It is briefly described below as we will refer to it to establish the Huffman tree characterization.
\subsection{The Huffman Algorithm} \label{subsec:Huffman_algorithm}
Let $P=\p1p$ be the given PMF for which we need an optimal code. Consider a forest $F_1$ containing $n$ isolated vertices. View each of the vertices as a tree and associate the probability $p_i$ with the $i$th vertex. Identify the integer $m$ such that $2\leq m \leq D$ and $D-1$ divides $n-m$. Choose any $m$ vertices $v_1,v_2,\ldots,v_m$ from $F_1$ having the lowest $m$ probabilities. Add a new vertex $r$ to $F_1$ and make it the parent of the vertices $v_1,v_2,\ldots,v_m$. Associate with the vertex $r$ the probability equaling the sum of the probabilities of its children. Finally, make the vertex $r$ the root of the tree to which it belongs.

Now suppose that the forests $F_1,F_2,\ldots,F_i$ are defined, with $F_i$ not being a tree. Define the forest $F_{i+1}$ as follows: choose any $D$ roots of the trees in $F_i$ with lowest $D$ probabilities. As in the previous case assign a parent to them, making the parent the root of the tree to which it belongs and assigning it the probability equaling the sum of the probabilities of its children. This is the forest $F_{i+1}$. From the way the first step of the algorithm was carried out, we will end with a forest that is in fact a tree associated with $P$; let us call it a \emph{Huffman tree} for the PMF $P$. 
\subsection{Huffman Tree Characterization}
Let $T$ be a tree for a PMF $P$. As done during the Huffman algorithm, let us associate each node of $T$ with a probability in the following way:
 a node gets the probability equal to the sum of the probabilities of its children. This assignment will result in the root node getting the probability 1. It can be shown that the expected length of $T$ is the sum of the probabilities of all its internal nodes (including the root node).
 
 Gallager \cite{Gallager} has given a necessary and sufficient condition for a tree associated with a PMF to be a \emph{Huffman tree}, i.e. a tree which can be generated by the application of Huffman algorithm on $P$. For our purpose, 
 we will need 
 a more descriptive version of this condition
 which
 is stated in the following theorem. A part of its proof follows the arguments presented in \cite{Gallager}.
%
 
 \begin{thm}\label{thm:sibling_property}
 Let $T$ be a tree associated with a PMF $P \in \P$. It is a Huffman tree for $P$  iff the following conditions hold:\\
 (P1) The probability of a lower level node is greater than or equal to that of a higher level node.\\
 (P2) Let $m$ be such that $D-1$ divides $n-m$ with $2 \leq m \leq D$. The tree $T$ contains a sibling set with $m$ nodes at its maximum level  which has the lowest $m$ probabilities of $P$. All the other sibling sets of $T$ have exactly $D$ nodes.\\
 (P3) The nodes at any level of $T$ can be listed in such a way that their probabilities are in a non-decreasing order, and the siblings come next to each other in the listing.
 \end{thm}
 
 \begin{proof}
 Let us first take $T$ to be a Huffman tree for $P$. 
 Now suppose that the condition (P1) is not satisfied for two nodes $v_1$ and $v_2$ at levels $l_1$ and $l_2$ respectively, with $l_1<l_2$. Let $v_1$ have the probability $p_1$ and $v_2$ have the probability $p_2$. Let $ST(v_1)$ ($ST(v_2)$) be the subtree of $T$
 that has $v_1$ ($v_2$) as its root and contains
  all the descendants of $v_1$ ($v_2$). Let $l_{11},l_{12},\ldots,l_{1i}$ be the levels of the leaves of 
 $ST(v_1)$, with levels calculated from the node $v_1$, 
 and $p_{11},p_{12},\ldots,p_{1i}$ be their respective probabilities (with $\sum_k p_{1k}=p_1$). Similarly, let $l_{21},l_{22},\ldots,l_{2j}$ be the levels of the leaves of 
 $ST(v_2)$, with levels calculated from the node $v_2$, 
and $p_{21},p_{22},\ldots,p_{2j}$ be their respective probabilities 
 (with $\sum_k p_{2k}=p_2$). Since $p_2>p_1$, 
 $v_2$ cannot be a descendant of $v_1$.
 Create a new tree associated with $P$ as follows without changing the association between the probabilities of $P$ and the leaves of $T$: Let $u_1$ and $u_2$ be the parents of $v_1$ and $v_2$ respectively. Delete the edges between $u_1$ and $v_1$, and $u_2$ and $v_2$. Construct new edges from $u_1$ to $v_2$, and  $u_2$ to $v_1$. In other words, we are interchanging the parents of $v_1$ and $v_2$. It is clear that the 
 resulting  graph is a tree. Let us call this tree $T'$.
  Let $L(T)$ and $L(T')$ denote the expected lengths of $T$ and $T'$ respectively.
  Now, there exists a $\lambda \in \mathbb{R}$ such that
\begin{align}
  L(T) &= \lambda + \sum_{k=1}^{i}(l_1 + l_{1k})p_{1k} + \sum_{k=1}^{j}(l_2 + l_{2k})p_{2k} \notag\\ 
  &= \lambda+ l_1p_1 + l_2p_2 + \sum_{k=1}^{i}l_{1k}p_{1k} + \sum_{k=1}^{j} l_{2k}p_{2k}, \label{eq:opt1}
  \end{align}
  and
  \begin{align}
  L(T') &= \lambda + \sum_{k=1}^{i}(l_1 + l_{2k})p_{2k} + \sum_{k=1}^{j}(l_2 + l_{1k})p_{1k} \notag\\
  &= \lambda + l_1p_2 + l_2p_1 + \sum_{k=1}^{i}l_{1k}p_{1k} + \sum_{k=1}^{j}l_{2k}p_{2k}. \label{eq:opt2}
  \end{align}
  Notice that $L(T')< L(T)$ which contradicts the optimality of $T$. Thus the condition (P1) is true.
 
 The condition (P2) follows from the way the Huffman algorithm is carried out and the condition (P1). 
 
 Let us now prove condition (P3). Let $F_1, F_2,\ldots, F_N$ be the forests obtained during the execution of the Huffman algorithm on PMF $P$ that yields the tree $T$ (see Section \ref{subsec:Huffman_algorithm}). Let $S_i$ ($1\leq i \leq N$) be the chosen root nodes in $F_i$ having the lowest $D$ probabilities (lowest $m$ probabilities when $i=1$). It can be seen that the $S_i$'s are precisely the sibling sets in $T$. Further, 
 each of the probabilities in $S_i$ is less than or equal to each of that in $S_{i+1}$. Thus, it is possible to list all the nodes of $T$ in such a way that their probabilities are in a non-decreasing order, and the siblings come next to each other in the listing. As a result, (P3) also holds.

 Let us prove that these conditions are sufficient for $T$ to be a Huffman tree for $P$. We will do so by showing that a tree isomorphic to $T$ can be obtained by the application of the Huffman algorithm on $P$.
 
 Let $V$ and $V_L$ denote the nodes and leaf nodes of $T$ respectively. Let $U$  be a set with $|U|=|V|$. Define a bijection $\psi:V\rightarrow U$. For any $V'\subseteq V$, let the $\psi$-image of $V'$ be the set $\{\psi(v')\mid v' \in V'\}$. 
 Construct a forest $F$ exactly containing all the elements of the $\psi$-image of $V_L$ as isolated vertices (with each of them viewed as a tree with that vertex itself serving as its root). 
 For each node of $V_L$, assign its probability to its $\psi$-image in the forest. Let $S$ be the sibling set in $T$ having the $m$ lowest probabilities of $P$, and let $r$ be its parent. 
 Derive a new forest $F'$ from $F$
 by introducing $\psi(r)$ as a new vertex to $F$, and making it the parent of the vertices occurring in the $\psi$-image of $S$. Assign the probability of the node $r$ to $\psi(r)$.
 Derive a new tree $T'$ from $T$ by deleting all the nodes of $S$ from $T$. It can be seen that the roots of the trees in $F'$ are precisely the leaf nodes of $T'$. 
 
 If $F'$ is not a tree, then perform the above mentioned step for $T'$ and $F'$ by taking $S$ to be a sibling set in $T'$ containing the lowest $D$ probabilities
 of $T'$. Such a choice is possible as $T'$ also satisfies conditions (P1) and (P3). Note that all the nodes of $S$ are leaf nodes of $T'$.
 The execution of the step will leave us with a tree and a graph both of which are derived in two steps from $T$ and $F$  respectively. The graph derived from $F$ will be a forest as the $\psi$-image of $S$ are root nodes in $F'$. 
 Keep repeating the above step and it can be seen that each execution of the step will yield a tree derived from $T$ and a forest derived from $F$ with the $\psi$-image of the leaves of the tree being exactly the root nodes of the trees of the forest. Continue doing so till the forest derived from $F$ becomes a tree, say $T^*$. It can be seen that these steps constitute the Huffman algorithm and so $T^*$ is a Huffman tree for $P$. By way of construction of $T^*$, it is clear that it is isomorphic to $T$ under the bijection $\psi$. Thus, $T$ is a Huffman tree for $P$. 
 \end{proof}
 \begin{remark} \label{rem:optimal}
Equations (\ref{eq:opt1}) and (\ref{eq:opt2}) can be used to show that if a tree $T$ is optimal, then it should at least satisfy the condition (P1).
 \end{remark}
 
 Let us call the conditions (P1)--(P3) as \textit{Huffman tree properties}. We will be referring to each of them  as follows: the condition (P1) will be called the \emph{level property}, condition (P2) will be called the \emph{maximum-level sibling property} and condition (P3) will be called the \emph{sibling property}.
  These properties serve as a potent tool to approach questions related to Huffman trees.
 
\section{Points of Maximum}
Let the sequence of the codeword lengths of a prefix code $X$ arranged in a non-decreasing order be called the \emph{length sequence} of $X$. Let an \emph{optimal length sequence} for a PMF $P$  be the length sequence of an optimal code for $P$. 
Let us follow the convention of always writing out the probabilities of a PMF $P=(p_1,p_2,\ldots,p_n)$ in a non-increasing order, i.e., $p_1\geq p_2\geq \ldots \geq p_n$.

Let us now take up the following problem: what is an optimal length sequence for a point of maximum? 
The length sequence of the Huffman code for the uniform distribution --- which is a point of maximum --- is known (see, for e.g., \cite{extra}). To emphasize the point that problems related to Huffman trees can be effectively handled using the Huffman tree properties, we will now use these 
to determine the length sequence of the Huffman code for $U_n$.

%
%

Consider the tree having $D^m$ leaves at level $m$, for some $m\in \mathbb{N}$. Let us denote it as $T_U(m)$. We have the following result which has appeared in  \cite{extra}.

\begin{lemma} \label{lemma:fat_tree}
Let $n=D^m$, for some $m\in \mathbb{N}$. A PMF $P\in \P$ with $P=\p1p$ has $T_U(m)$ as a Huffman tree iff the sum of the lowest $D$ probabilities of $P$ is greater than or equal to its highest probability, i.e., iff  $$\sum_{i=n-D+1}^{n}p_i \geq p_1.$$
\end{lemma}
\begin{proof}
First suppose that $T_U(m)$ is a Huffman tree for $P$. 
The minimum probability of the nodes at level $m-1$ is 
$\sum_{i=n-D+1}^{n}p_i$ (maximum-level sibling property (P2)).
From the level property (P1), we have that $\sum_{i=n-D+1}^{n}p_i \geq p_1$.

Now suppose that $\sum_{i=n-D+1}^{n}p_i \geq p_1$. Assign the probabilities of $P$ to the leaves of $T_U(m)$ in such a way that the probabilities of the leaves from left to right are in a non-decreasing order. We will now show $T_U(m)$ with this assignment of probabilities satisfies the Huffman tree properties. We will do it by induction on $m$. The statement is clearly true for $m=1$. Let us assume that for some positive integer $k$, the statement is true for $m=k$. Let us now take $m=k+1$. Let us denote the PMF formed out of the probabilities at level $k$ of $T_U(k+1)$ to be $Q$. Observe that the probabilities of $Q$ are arranged in a non-decreasing order from left to right at the $k$th level.
Let $q_1,q_2,\ldots,q_D$ be the lowest $D$ probabilities of $Q$ taken in a non-decreasing order. The highest probability of $Q$, say $q_{h}$, is given by $p_{1}+p_{2}+\cdots + p_D$. Since $q_1=\sum_{i=n-D+1}^{n}p_i$, we have the following inequalities:
\begin{equation*}
q_D\geq q_{D-1}\geq \ldots \geq q_1 \geq p_1 \geq p_{2}\geq \ldots \geq  p_{D}.
\end{equation*}
 Thus, we have that $\sum_{i=1}^{D}q_i\geq q_h$. By the induction hypothesis, we have that the tree $T_U(k)$, with its nodes retaining its probabilities as in $T_U(k+1)$, satisfies the Huffman tree properties. 
 Now take a node at level $k$ in the tree $T_U(k+1)$ with probability $q_i$ and a node at level $k+1$ with probability $p_j$. We have that $q_i \geq q_1 \geq p_1 \geq p_j$. Thus, the \level (P1) is satisfied in $T_U(k+1)$. Finally, by the way the probabilities were assigned to the leaves of $T_U(k+1)$, the \maxsib (P2) and the \sib (P3) are also satisfied at level $k+1$  of $T_U(k+1)$. Thus, from Theorem \ref{thm:sibling_property}, the tree $T_U(k+1)$ is a Huffman tree for $P$. 
Hence the lemma is proved.
\end{proof}

\begin{remark}
For the case $n=D^m$ ($m \in \mathbb{N}$), Lemma \ref{lemma:fat_tree} describes points of maximum, other than the uniform distribution $U_n$, for the minimum expected length function $\hL$. 
For a PMF $P\in \P$, satisfying the condition of Lemma \ref{lemma:fat_tree}, i.e., $\sum_{i=n-D+1}^{n}p_i \geq p_1$, we have from this lemma that $\hL(P)=m$. Since $U_n$ also satisfies this condition of the lemma, and since it is a point of maximum, we have that $P$ is also a point of maximum. Thus, all the PMFs satisfying the condition of Lemma \ref{lemma:fat_tree} are points of maximum.
\end{remark}

Now, let $T_U$ represent the tree $T_U(m)$  when $n=D^m$ (a power of $D$) and when $D^m < n< D^{m+1}$ ($n$ is not a power of $D$) let $T_U$ denote the tree (see Fig.~\ref{teeu}) in which
%
\begin{itemize}
\item[i)] all the leaves are at levels $m$ and $m+1$,
\item[ii)] at level $m$, each of the internal nodes is to the right of any of the leaf nodes and
\item[iii)] all the internal nodes, except possibly the leftmost  internal node at level $m$, have $D$ children each. The leftmost internal node at level $m$ has at least 2 children.
\end{itemize}

\begin{figure}
  \centering
   \resizebox{7cm}{!}{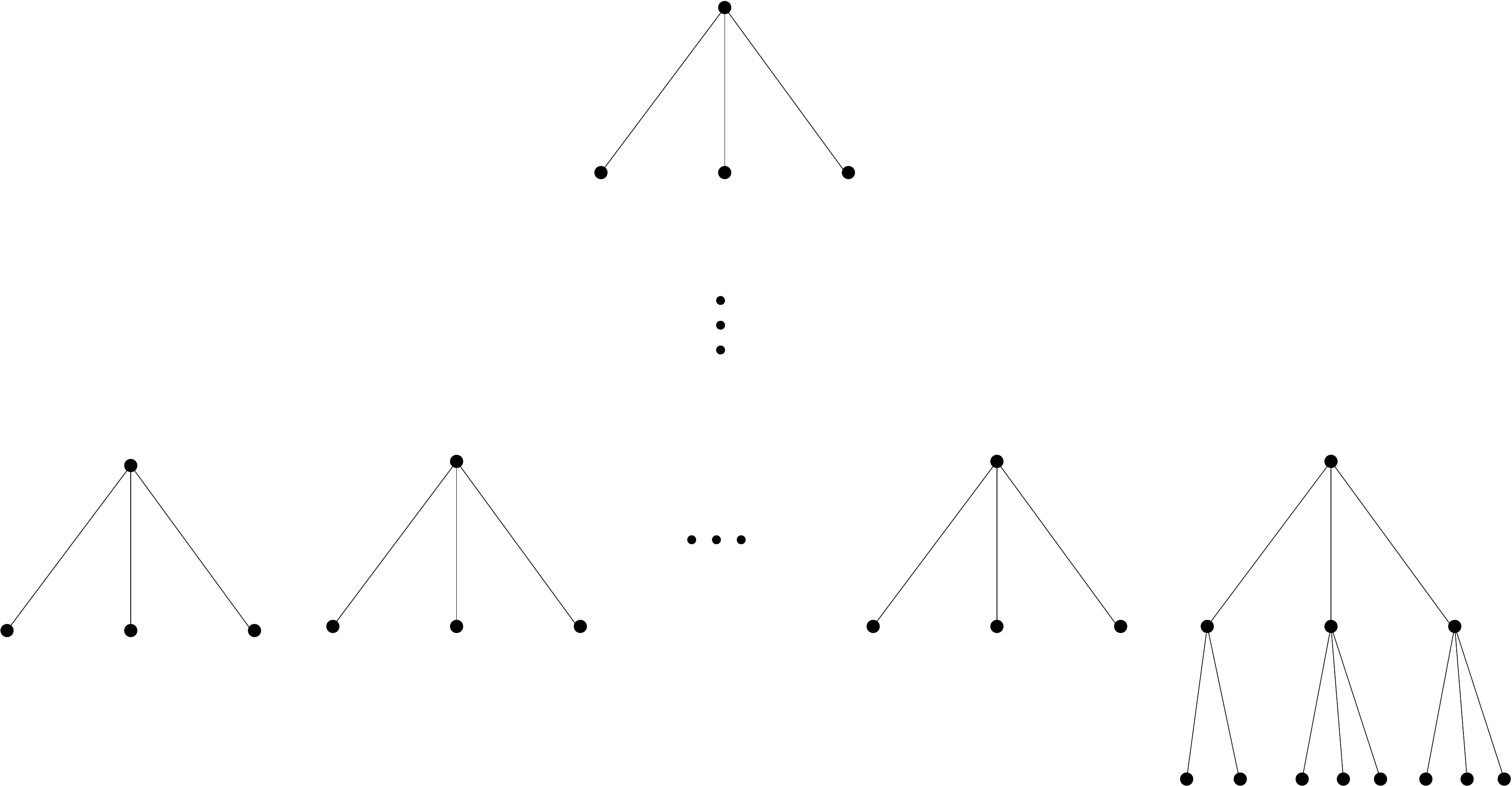}
    \caption{An example for $T_U$ when $D=3$ and $n$ is not a power of $D$.} 
\label{teeu}
\end{figure}
It can be seen that these conditions uniquely define the tree $T_U$ when $n$ is not a power of $D$. When $D-1$ divides $n-1$, then all the internal nodes of $T_U$ will have $D$ children; and the number of internal nodes at level $m$ is $(n-D^m)/(D-1)$. When $D-1$ doesn't divide $n-1$, only one internal node will have $m$ children where $m$ is such that $2\leq m \leq D$ and $D-1$ divides $n-m$; and the number of internal nodes at level $m$ is $\lceil (n-D^m)/(D-1) \rceil$.

Now we have the following result.

\begin{thm}\label{thm:Huffman:uniform_dis_opt}
The tree $T_U$ is a Huffman tree for the uniform distribution $U_n$.
\end{thm}

\begin{proof}
Lemma \ref{lemma:fat_tree} tells us that the theorem is true when 
$n$ is a power of $D$.
Let us take that $n$ satisfies $D^m < n< D^{m+1}$. Consider $T_U$ and assign the probabilities of $U_n$ to its leaves. Consider the tree $T_U(m)$ obtained by removing all the leaves from $T_U$ at level $m+1$. Let the nodes of $T_U(m)$ retain their probabilities as they were in $T_U$. Now look at the probabilities of the nodes of $T_U(m)$ at level $m$. The lowest probability is at least $1/n$ and the highest probability is at most $D/n$.
Thus, by Lemma \ref{lemma:fat_tree} and Theorem \ref{thm:sibling_property}, we have that $T_U(m)$ satisfies the Huffman tree properties. This, along with the way the tree $T_U$ is defined shows us that $T_U$ also satisfies the Huffman tree properties. Thus, 
$T_U$ is a Huffman tree for $U_n$.
\end{proof}
Thus, we have that when $n=D^m$, an optimal length sequence for $U_n$ is $(m,m,\ldots,m)$ and when $D^m<n<D^{m+1}$ an optimal length sequence for $U_n$ is $(m,m,\ldots,m,m+1,m+1,\ldots,m+1)$ with $D^m- \lceil (n-D^m)/(D-1)\rceil$ occurrences of $m$. Let us denote this optimal length sequence for $U_n$ by $\sfL_U$.


It turns out that Hwang's argument, as in \cite{Hwang}, can now be used
to determine an optimal length sequence for any point of maximum. 
 Let $P=(p_1,p_2,\ldots,p_n)$ and $Q=(q_1,q_2,\ldots,q_n)$ be two PMFs from $\P$. If the sum of the highest $k$ probabilities of $P$ is greater than or equal to the sum of the highest $k$ probabilities of $Q$, i.e. if $\sum_{i=1}^{k}p_i \geq \sum_{i=1}^{k}q_i$, for all $k$ satisfying $1\leq k \leq n$, then $P$ is said to \emph{majorize} $Q$
and is denoted as $P\succ Q$. Note that this relation is a partial order. It is well-known that every $P\in \P$ majorizes the uniform distribution $U_n$.

\begin{lemma}\label{lemma:Hwang_modified}
If $P$ is a point of maximum, then the following hold:\\
i) $\sfL_U$ is an optimal length sequence for $P$. \\
ii) When $n$ is not a power of $D$, the PMF $P$ is the uniform distribution $U_n$, or, in other words, $U_n$ is the unique point of maximum.
\end{lemma}
\begin{proof}
Let 
$P=(p_1,p_2,\ldots,p_n)$ 
and let $\sfL_{P}=(\le{p}_1,\le{p}_2,\ldots,\le{p}_n)$ be an optimal length sequence for $P$. From our ordering convention, we have that $\le{p}_i$ is the length of the codeword associated with $p_i$. 
Let us also consider the optimal length sequence $\sfL_U$ for $U_n$ and write it out as $(\le{u}_1,\le{u}_2,\ldots,\le{u}_n)$. 
Take $\le{u}_{0}=\le{p}_{0}=0$.
Making use of Hwang's technique,
we get
 the following chain of inequalities:
   \begin{align}
     \hL(U_n) &= \sum_{i=1}^{n} \le{u}_{i}/n,\notag \\
        &= \sum_{i=1}^{n} (\le{u}_{i}-\le{u}_{i-1}) \sum_{j=i}^{n}\frac{1}{n}, \label{peq:subeq_mod_Hwang0} \\
        &\geq  \sum_{i=1}^{n} (\le{u}_{i}-\le{u}_{i-1}) \sum_{j=i}^{n}p_{j}, \label{peq:subeq_mod_Hwang1} \\
        &=  \sum_{i=1}^{n} \le{u}_{i} p_{i},\notag\\
        &\geq \sum_{i=1}^{n} \le{p}_{i} p_{i} , \label{peq:subeq_mod_Hwang2}\\
        &= \hL(P),\notag
    \end{align}
  where (\ref{peq:subeq_mod_Hwang1}) follows from the relation $P\succ U_n$, and (\ref{peq:subeq_mod_Hwang2}) follows from the fact that $\sfL_P$ is an optimal length sequence for $P$. Since $\hL(U_n)= \hL(P) $, we have that the inequalities in (\ref{peq:subeq_mod_Hwang1}) and in (\ref{peq:subeq_mod_Hwang2}) are equalities. Let us see what they imply.
  
 i) Since the inequality in (\ref{peq:subeq_mod_Hwang2}) is actually an equality, we have that $\sfL_U$ is an optimal length sequence for $P$.
  
  ii) Let $n$ be such that $D^m <n< D^{m+1}$. Since the entries in the sequence $\sfL_U$ are either $m$ and $m+1$, the expression in (\ref{peq:subeq_mod_Hwang0}) boils down to $m\sum_{j=1}^{n}1/n+ \sum_{j=k+1}^{n} 1/n$, for $k$ such that $\le{u}_k=m$ and $\le{u}_{k+1}=m+1$. Similarly, the expression in (\ref{peq:subeq_mod_Hwang1}) is $m\sum_{j=1}^{n}p_j+ \sum_{j=k+1}^{n} p_j$. Since the inequality in (\ref{peq:subeq_mod_Hwang1}) is an equality, we have that 
   \begin{equation}
   \frac{n-k}{n}= \sum_{j=k+1}^{n} p_j. \label{peq:pom_avg}
   \end{equation} 
    Let $I_1=\{1,2,\ldots,k\}$ and $I_2=\{k+1,k+2,\ldots,n\}$. We have that $p_{k+1}$ is greater than or equal to the average of $\{p_j\}_{j\in I_2}$, which equals $1/n$ (from \ref{peq:pom_avg}). Thus, $p_{k+1}\geq 1/n$. From  (\ref{peq:pom_avg}), we also have that

\begin{equation}
\sum_{j\in I_1}p_j=k/n. \label{leq:I1_avg}
\end{equation}

For $j\in I_1$, since $p_j\geq p_{k+1}\geq 1/n$, (\ref{leq:I1_avg}) can occur iff 
    \begin{equation}
p_j=1/n,\text{  for all }j\in I_1. \label{peq:pom1}
    \end{equation}
    
    Now, if there exists a $j\in I_2$ such that $p_j<1/n$, then  (\ref{peq:pom_avg}) implies that $p_{k+1}>1/n$. But this contradicts (\ref{peq:pom1}). Thus, $P$ is the uniform distribution.
    Hence, when $n$ is not a power of $D$, uniform distribution $U_n$ is the unique point of maximum.
\end{proof}

Lemma \ref{lemma:Hwang_modified} tells that when $n$ is not a power of $D$, the uniform distribution is the unique point of maximum. It also tells that when 
$n$ is a power of $D$, 
the set of all the points of maximum is precisely the collection of
all the PMFs that have $\sfL_U$ as their optimal length sequence. Is this set the same as the set of all PMFs whose Huffman code's length sequence is $\sfL_U$? 
The answer to this question is not clear at this stage
because the Huffman algorithm does not generate all the optimal codes. The following result will come to our aid:

\begin{lemma}\label{lemma:optimal_Huffman}
 {If $\sfL$ is an optimal length sequence for a $P\in \P$, then} there exists a Huffman tree for $P$ with $\sfL$ as its length sequence.
\end{lemma}

\begin{proof}
 Let $T$ be an optimal tree for $P$ with $\sfL$ as its length sequence.
 It should at least satisfy the \level (P1) (Remark \ref{rem:optimal}).
 
Let $l_{\max}$ be the maximum level of a node of $T$.
Change, if necessary, the assignment of children at level $l_{\max}$ to the internal nodes at level $l_{\max}-1$ so that the resulting tree is still a subtree of $T_{\infty}$, but now the \maxsib (P2) and the \sib (P3) are satisfied by the nodes at level $l_{\max}$. This will not create any new leaf nodes at level $l_{\max}-1$ of the resulting tree, call it $T'$, for otherwise one of the leaves at level $l_{\max}$ can be deleted and its probability can be assigned to one of the newly created leaves at level $l_{\max}-1$. If we now throw away the other leaves, if present, without any assignment of probabilities, then the expected length of this new tree is strictly less than that of $T$ which is not possible. Further, we have that $T'$ is optimal.
 
  Assume that we now have a tree
  \begin{itemize}
  \item[i)] which is optimal, and
  \item[ii)]  in which the \sib (P3) is true for the nodes at levels $l_{\max}-k$ to $l_{\max}$, for some $k$ satisfying $0 \leq k < l_{\max}-1$.
  \end{itemize}
  
Change, if necessary, the assignment of children at level $l_{\max}-k-1$ to the internal nodes at level $l_{\max}-k-2$ so that the resulting tree is still $D$-ary, but now  the \sib (P3) is satisfied by the nodes at level $l_{\max}-k-1$. A tree obtained after this re-assignment is clearly optimal. Moreover, the nodes at levels $l_{\max}-k$ to $l_{\max}$ still satisfy the \sib (P3)
    as the nodes at these levels which were siblings before the rearrangement remain so after it. This process can be continued till we get a tree which is optimal and satisfies (P2) and (P3). This tree is a Huffman tree for $P$ (Remark \ref{rem:optimal} and Theorem \ref{thm:sibling_property}) and has $\sfL$ as its length sequence.
 \end{proof}
 
 Even though the Huffman algorithm is restrictive in the range of optimal trees it constructs, Lemma \ref{lemma:optimal_Huffman} assures us that the algorithm can yield an optimal tree corresponding to any optimal length sequence.
 
 Thus, for the case of $n=D^m$, any point of maximum has a Huffman tree with $\sfL_U$ as its length sequence. From Lemmas \ref{lemma:fat_tree} and \ref{lemma:Hwang_modified}, the following theorem follows:
  
  \begin{thm}\label{thm:Huffman_main}
  i) When $n$ is a power of $D$, the points of maximum for the minimum expected length function $\hL$ are those $P=\p1p$ from the set $\P$
  for which the sum of the lowest $D$ probabilities is greater than or equal to its maximum probability, i.e. 
   $$\sum_{i=n-D+1}^{n}p_i \geq p_1.$$\\
   ii) When $n$ is not a power of $D$, the uniform distribution $U_n$
  is the unique point of maximum for $\hL$.  
  \end{thm} 

\section{Concluding Remarks}
The points of maximum of the minimum expected length function have been completely characterized using a characterization of Huffman trees, and a chain of inequalities that Hwang used to show 
the Schur-concavity of the minimum expected length function. This result shows that the points of maximum known in the literature are all that exist.


\end{document}